\documentclass[10pt,conference]{IEEEtran}
\IEEEoverridecommandlockouts

\usepackage{misc/mypaper}
\usepackage{misc/mymacros}

\title{Track-To-Track Association for Fusion of Dimension-Reduced Estimates\thanks{This work has been supported by the Industry
    Excellence Center LINK-SIC funded by The Swedish Governmental
    Agency for Innovation Systems (VINNOVA) and Saab AB, and by the project Scalable Kalman filters funded by the Swedish Research Council (VR). G.~Hendeby has received funding from the Center for Industrial Information Technology at Link\"{o}ping University (CENIIT) grant no. 17.12.}}

\author{\IEEEauthorblockN{%
    Robin Forsling\IEEEauthorrefmark{1}\IEEEauthorrefmark{2},
    Zoran Sjanic\IEEEauthorrefmark{1}\IEEEauthorrefmark{2},
    Fredrik Gustafsson\IEEEauthorrefmark{1}, and
    Gustaf Hendeby\IEEEauthorrefmark{1}}
  \IEEEauthorblockA{\IEEEauthorrefmark{1}
    Dept. of Electrical Engineering,
    Link\"{o}ping University,
    Link\"{o}ping, Sweden\\
    e-mail: \texttt{\url{{firstname.lastname}@liu.se}}}
  \IEEEauthorblockA{\IEEEauthorrefmark{2}
    Saab AB,
    Link\"{o}ping, Sweden\\%
    e-mail: \texttt{\url{{firstname.lastname}@saabgroup.com}}}
}

\begin{document}

\maketitle

\IEEEpubid{\begin{minipage}{\textwidth}\ \\[60pt] 
Accepted to be published in \emph{Proceedings of the 26th IEEE International Conference on Information Fusion}, Charleston, SC, USA, Jun. 2023. \\
\copyright\xspace 2023 IEEE. Personal use of this material is permitted. Permission from IEEE must be obtained for all other uses, in any current or future media, including reprinting/republishing this material for advertising or promotional purposes, creating new collective works, for resale or redistribution to servers or lists, or reuse of any copyrighted component of this work in other works.
\end{minipage}} 

\begin{abstract}
Network-centric multitarget tracking under communication constraints is considered, where dimension-reduced track estimates are exchanged. Previous work on target tracking in this subfield has focused on fusion aspects only and derived optimal ways of reducing dimensionality based on fusion performance. In this work we propose a novel problem formalization where estimates are reduced based on association performance. The problem is analyzed theoretically and problem properties are derived. The theoretical analysis leads to an optimization strategy that can be used to partly preserve association quality when reducing the dimensionality of communicated estimates. The applicability of the suggested optimization strategy is demonstrated numerically in a multitarget scenario.
\end{abstract}



\section{Introduction} \label{sec:intro}

The {multitarget tracking} (MTT, \cite{vo2015wiley:mtt}) problem is a well-studied topic. Two popular classical MTT methods are the {global nearest neighbor} (GNN) tracker and the {multiple hypothesis tracker} \cite{reid1979:mht,chong2018:mht}. In the last few decades different MTT methods based on {random finite sets} have emerged that provide a solid mathematical framework for the {multitarget Bayesian filter}, see, \eg, \cite{mahler2003:taes,granstrom2018fusion:pmbm}. A key feature of all of these MTT algorithms is how they deal with the \emph{association problem} where measurements are assigned to existing tracks. Association problems also arise in \emph{network-centric} MTT where multiple agents estimate a common set of targets and the communicated tracks must be associated with local tracks. This is a \emph{track-to-track association} problem. In addition, the communication channel is a limited resource and in certain situations the exchanged data must be reduced \cite{kimura:survey-data-comp,forsling2019}, which in general have a negative impact on the association quality. 

A network-centric MTT scenario with dimension-reduced estimates is illustrated in Fig.~\ref{fig:intro-scenario}. The problem of {fusing} dimension-reduced measurements and estimates have been studied before: In \cite{zhang2003,fang2010tsp,msechu2012tsp} it is done in centralized and distributed configurations, and in \cite{forsling2022fusion,forsling23tsp-gevo,forsling23aero} it is done for decentralized sensor networks. However, all of these papers assume that the association process, see Fig.~\ref{fig:intro-system}, can be neglected or is trivially solved such that the dimension-reduction can be optimized for fusion performance only. The corresponding {association} problem---where the dimension-reduction takes data association into account---remains untreated.

\begin{figure}[tb]
	\centering
	\begin{subfigure}[b]{1.0\columnwidth}
		\centering
		\begin{tikzpicture}[scale=.23]

\def\rc{0.8}
\def\lx{2.75}
\def\ly{1.75}

\begin{footnotesize}

\def\xa{(24,0)}
\draw [rotate around={140:\xa}] \sensoroption{colora} (24,0) -- (28,3) to [out=-60,in=60] (28,-3) -- (24,0);	
\draw \agentoption{colora} \xa circle [radius = \rc];
\node at \xa {1};

\def\xa{(-4,0)}
\draw [rotate around={45:\xa}] \sensoroption{colorb} (-4,0) -- (0,3) to [out=-60,in=60] (0,-3) -- (-4,0);
\draw \comoption (-4,0) -- (3.5,0) node [right,align=center] {\scriptsize{dimension-reduced estimates}};
\draw \agentoption{colorb} \xa circle [radius = \rc];
\node at \xa {2};

\end{footnotesize}

\node at (2,10) {\mycircle{black}};
\node at (9,5) {\mysquare{black}};
\node at (16,9) {\myuptriangle{black}};

\def\xe{(2.2,11)}
\node at \xe {\mycircle{colora}};
\draw \ellipseoption{colora} \xe ellipse [x radius = \lx, y radius = \ly, rotate = -30];

\def\xe{(8.5,4.5)}
\node at \xe {\mysquare{colora}};
\draw \ellipseoption{colora} \xe ellipse [x radius = \lx, y radius = \ly, rotate = -15];

\def\xe{(17,9.5)}
\node at \xe {\myuptriangle{colora}};
\draw \ellipseoption{colora} \xe ellipse [x radius = \lx, y radius = \ly, rotate = -60];

\def\xe{(2.5,9.5)}
\node at \xe {\mycircle{colorb}};
\draw \ellipseoption{colorb} \xe ellipse [x radius = \lx, y radius = \ly, rotate = 60];

\def\xe{(9.0,5.75)}
\node at \xe {\mysquare{colorb}};
\draw \ellipseoption{colorb} \xe ellipse [x radius = \lx, y radius = \ly, rotate = 20];

\def\xe{(16.5,8.25)}
\node at \xe {\myuptriangle{colorb}};
\draw \ellipseoption{colorb} \xe ellipse [x radius = \lx, y radius = \ly, rotate = 20];
		\end{tikzpicture}	
		\caption{Agent~2 transmits dimension-reduced estimates to agent~1.}	
		\label{fig:intro-scenario}	
	\end{subfigure}
	\begin{subfigure}[b]{1.0\columnwidth}
		\centering
		\begin{tikzpicture}[scale=.59]
			\colorlet{colorbox}{colorb!40!}
\colorlet{colorscope}{myorange!110!}

\def\lineopt{[thick,-latex]}
\def\boxopt{[ultra thick,rounded corners,fill=white]}

\draw [white,opacity=0] (0,0) rectangle (1,1.25);

\begin{footnotesize}

\draw [thick,black,fill=colora!40!,rounded corners] (0,-2.5) rectangle (6.5,1);
\draw [thick,black,fill=colorb!40!,rounded corners] (-8,-2.5) rectangle (-2,1);

\draw[dashed,rounded corners,ultra thick,colorscope,fill=colorscope,fill opacity=0.3](-7.75,-0.75) rectangle (3.25,0.75);

\draw\lineopt (-3.5,0)--(0.5,0);
\draw\boxopt (-7.5,-0.5) rectangle (-2.5,0.5);
\node at (-5,0) {\bfseries dimension-reduction};

\draw\lineopt (-5,-2.0)--(-5,-0.5) node [near start,left,align=center] {agent~2 \\ estimates};

\draw\boxopt (0.5,-0.5) rectangle (3,0.5);
\node at (1.75,0) {\bfseries association};

\draw\lineopt (3,0)--(4,0);

\draw\boxopt (4,-0.5) rectangle (6,0.5);
\node at (5,0) {\bfseries fusion};

\draw\lineopt (2,-2.0)--(2,-0.5) node [near start,right,align=center] {agent~1 \\ estimates};


\end{footnotesize}
		\end{tikzpicture}	
		\caption{Schematics of the considered multitarget tracking problem.}	
		\label{fig:intro-system}
	\end{subfigure}
	\caption{A multiagent multitarget tracking scenario where agent~2 transmits dimension-reduced estimates to agent~1. The colored numerated circles in (a) represent agents. The black symbols represent targets and the corresponding colored symbols and ellipses are estimates. Before fusing received dimension-reduced estimates, agent~1 must associate these estimates with its local estimates. The scope of this paper is highlighted by the red dashed box in (b).}
	\label{fig:intro}
\end{figure}
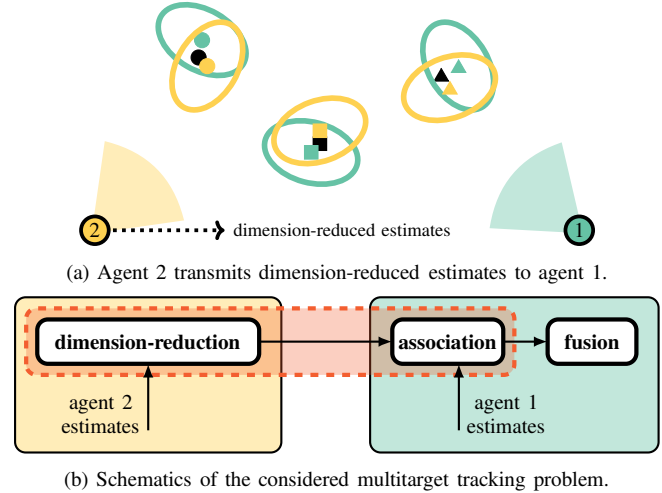

In this paper we deal with the association problem in network-centric MTT with dimension-reduced estimates. The main goal is to find a way to compute the dimension-reduction such that satisfactory association performance is obtained. This problem is formalized and the relationship to fusion optimal dimension-reduction is discussed. As a result of a problem analysis an optimization strategy is suggested for computing dimension-reductions that yield good association performance. The contributions are listed below.

\begin{itemize}
	\item We propose a novel formalization of the association problem in network-centric MTT with dimension-reduced estimates. This problem formulation essentially involves a GNN tracker and computation of the dimension-reduction such that satisfactory association quality is obtained.
	\item The proposed problem is analyzed theoretically and problem properties are derived.
	\item Based on the problem analysis we suggest an optimization algorithm for the dimension-reduction computation. 
\end{itemize}	

\section{Network-Centric Target Tracking Using Dimension-Reduced Estimates} \label{sec:problem}

In this section we introduce the studied multiagent MTT problem. The outlined estimation model forms the basis for fusion related operations and is mainly related to previous work. The provided association model is fundamental for the contributions of this paper. We give a motivating example to illustrate why association properties should be taken into account when reducing dimensionality. The considered problem is formalized at the end.

\subsection{Preliminaries} \label{sec:prel}
Let $\realsn$ and $\realsmn$ denote the set of all real-valued $n$-dimensional vectors and the set of all real-valued $m\times n$ matrices, respectively. Let $\psdsetn$ and $\pdsetn$ denote the set of all symmetric positive semidefinite $n\times n$ matrices and the set of all symmetric positive definite $n\times n$ matrices, respectively.

Targets and estimates are distinguished by subscript $(i)$, \eg, the state of the \ith target is $x\subi\in\realsn$. We use boldface to express random variables and normal face for a realization of the random variable, \eg, $y$ is a realization of $\bfy$. The expectation operator is denoted by $\EV(\cdot)$. A random variable $\bfy$ is said to be {Gaussian distributed} with mean $\mu=\EV(\bfy)$ and covariance matrix $\Sigma=\EV(\bfy-\mu)(\bfy-\mu)\trnsp$ if $\bfy\sim\calN(\mu,\Sigma)$.  


\subsection{Estimation Model} \label{sec:est-model}
We consider two agents. Let 
\begin{subequations} \label{eq:est-model}
\begin{align}
	y\ai &= x\subi + v\ai, & \bfv\ai &\sim \calN(0,R\ai), \\
	y\bi &= x\subi + v\bi, & \bfv\bi &\sim \calN(0,R\bi),
\end{align}	
\end{subequations}
be the local estimates of $x\subi$ in agent~1 and agent~2, respectively. For instance, $y\ai$ is the state estimate and $R\ai$ the corresponding covariance of the \ith target in agent~1. All cross-covariances $R_{12(i)}=\EV(\bfv\ai\bfv\bi\trnsp)$ are assumed to be zero\footnote{In network-centric MTT estimates are typically correlated to some degree. Here it is assumed that estimates have been decorrelated before they are communicated, for instance by using the techniques in \cite{tian-bar-shalom:gimf,forsling23aero}.}. A dimension-reduced estimate is given by
\begin{align}
	y\Mi &= \M\subi y\bi, & R\Mi &= \M\subi R\bi\M\subi\trnsp,
\end{align}	
where $\M\subi\in\realsmn$ with $m<n$ and $\rank(\M\subi)=m$. 

The sets of local estimates of agent~1 and agent~2 are 
\begin{subequations} \label{eq:est-set-nu}
\begin{align}
	\estset_1 &= \left\{(y\aa,R\aa),\dots,(y\aN,R\aN)\right\}, \\
	\estset_2 &= \left\{(y\ba,R\ba),\dots,(y\bN,R\bN)\right\}.
\end{align}	
\end{subequations}
Agent~1 and agent~2 track exactly the same targets and hence have the same number of tracks. Moreover, it is assumed that the elements of $\estset_1$ and $\estset_2$ are labeled according to $x\suba,\dots,x\subN$, \eg, $(y\ai,R\ai)$ and $(y\bi,R\bi)$ are estimates of the same target $x\subi$. This might sound a bit counterintuitive but the assumption is not a restriction since here the actual correct association result is assumed to be known, as described later, and the task is to compute $\M\suba,\dots,\M\subN$. We also define
\begin{equation}
	\estset_\M = \left\{(y\Ma,R\Ma),\dots,(y\MN,R\MN)\right\}.
	\label{eq:est-set-psi}
\end{equation}

Since $R_{12(i)}=0$, $(y\ai,R\ai)$ and $(y\Mi,R\Mi)$ are {mean square error} (MSE) optimally fused according to \cite{forsling2022fusion}
\begin{subequations} \label{eq:kfus}
\begin{align}
	\xhat\subi &= P\subi\left( R\ai\inv y\ai + \Mt\subi R\Mi\inv y\Mi \right), \\
	P\subi &= \left( R\ai\inv + \Mt\subi R\Mi\inv \M\subi \right)\inv.
\end{align}	
\end{subequations}
This fusion rule is denoted {Kalman fuser} (KF). For KF, a {fusion optimal}\footnote{Fusion optimal in the sense that this $\M\subi$ yields the smallest MSE when fusing $(y\ai,R\ai)$ and $(y\Mi,R\Mi)$.} $\M\subi$ is computed using Algorithm~\ref{alg:gevo-kf} \cite{forsling2022fusion}.

\begin{algorithm}[t]
	\caption{Fusion Optimal $\M\subi$}
	\label{alg:gevo-kf}
	\begin{algorithmic}[0]
		\begin{small}
		\Input $R\ai,R\bi\in\pdsetn$ and $m$ \\
		\begin{enumerate}[1:]
			\item Let $Q\subi = R\ai^2$ and $S\subii = R\ai+R\bi$. 
			\item Compute $\lambda_1\leq\dots\leq\lambda_n$ and $z_1,\dots,z_n$ as the solution to
			\begin{equation*}
				Q\subi z = \lambda S\subii z,
			\end{equation*}
			\item Compute $V=\BBM v_1&\dots&v_m\EBM$, where $v_a\trnsp v_b=\delta_{ab}$, such that $z_n,\dots,z_{n-m+1}$ span the same subspace as the columns of $V$.
			\item Compute $V\trnsp R_2V=U\Sigma U\trnsp$ and let $\M\subi=U\trnsp V\trnsp$.
		\end{enumerate}
		\Output $\M\subi$
		\end{small}
	\end{algorithmic}
\end{algorithm}

\begin{figure*}[t]
	\centering
	\begin{subfigure}[b]{1.0\columnwidth}
		\centering
		\begin{tikzpicture}[scale=.58]
			\input{fig/example-motivation.tex}
		\end{tikzpicture}	
		\caption{Target tracking example.}	
		\label{fig:ex-motiv-scenario}	
	\end{subfigure}
	\begin{subfigure}[b]{1.0\columnwidth}
		\centering
		\begin{tikzpicture}[scale=.82]
			\colorlet{colorarea}{myorange}

\newcommand{\curveoption}[1]{[line width=0.3mm,color=#1]}
\newcommand*\borderoption{[line width=0.1mm,color=colorarea,densely dashed]}
\renewcommand*\filloption{[colorarea!30!]}

\def\alo{63.438}; \def\ahi{116.562}

\def\xmin{0}; \def\xmax{190}
\def\xlbl{-12}

\begin{footnotesize}

\def\ymin{1.0}; \def\ymax{2.8}

\begin{scope}[xscale=0.035,yscale=1.4]
\fill \filloption (\alo,\ymin) rectangle (\ahi,\ymax);
\draw \supportopt{black} (90,1.2955)--(90,\ymin) node [below] {$\alpha^\star$};
\draw \coordframeopt (\xmin,\ymax)--(\xmin,\ymin)--(\xmax,\ymin) node [below] {$\alpha$};
\draw \curveoption{darkgray} plot [smooth] coordinates {( 0.0,2.545)( 5.0,2.536)(10.0,2.508)(15.0,2.462)(20.0,2.399)(25.0,2.322)(30.0,2.233)(35.0,2.134)(40.0,2.029)(45.0,1.920)(50.0,1.812)(55.0,1.707)(60.0,1.608)(65.0,1.519)(70.0,1.442)(75.0,1.379)(80.0,1.333)(85.0,1.305)(90.0,1.295)(95.0,1.305)(100.0,1.333)(105.0,1.379)(110.0,1.442)(115.0,1.519)(120.0,1.608)(125.0,1.707)(130.0,1.812)(135.0,1.920)(140.0,2.029)(145.0,2.134)(150.0,2.233)(155.0,2.322)(160.0,2.399)(165.0,2.462)(170.0,2.508)(175.0,2.536)(180.0,2.545)};
\node [above] at (180,2.545) {$\trace(P\suba)$};
\node at (90,2.5) {$J_0>J_e$};
\node at (31,1.25) {$J_0<J_e$};
\node at (149,1.25) {$J_0<J_e$};
\node [rotate=90] at (\xlbl,1.9) {Fusion loss};
\end{scope}

\def\ymin{0.0}; \def\ymax{14}

\begin{scope}[xscale=0.035,yscale=0.175,yshift=-9cm]
\fill \filloption (\alo,\ymin) rectangle (\ahi,\ymax);
\draw \coordframeopt (\xmin,\ymax)--(\xmin,\ymin)--(\xmax,\ymin) node [below] {$\alpha$};	
\draw \curveoption{darkgray} plot [smooth] coordinates {( 0.0,0.182)( 5.0,0.203)(10.0,0.264)(15.0,0.365)(20.0,0.501)(25.0,0.669)(30.0,0.864)(35.0,1.079)(40.0,1.309)(45.0,1.545)(50.0,1.782)(55.0,2.012)(60.0,2.227)(65.0,2.422)(70.0,2.590)(75.0,2.726)(80.0,2.827)(85.0,2.888)(90.0,2.909)(95.0,2.888)(100.0,2.827)(105.0,2.726)(110.0,2.590)(115.0,2.422)(120.0,2.227)(125.0,2.012)(130.0,1.782)(135.0,1.545)(140.0,1.309)(145.0,1.079)(150.0,0.864)(155.0,0.669)(160.0,0.501)(165.0,0.365)(170.0,0.264)(175.0,0.203)(180.0,0.182)};
\draw \curveoption{darkgray,dashed} plot [smooth] coordinates {( 0.0,11.818)( 5.0,11.728)(10.0,11.462)(15.0,11.027)(20.0,10.436)(25.0,9.707)(30.0,8.864)(35.0,7.930)(40.0,6.935)(45.0,5.909)(50.0,4.883)(55.0,3.888)(60.0,2.955)(65.0,2.111)(70.0,1.382)(75.0,0.792)(80.0,0.356)(85.0,0.090)(90.0,0.000)(95.0,0.090)(100.0,0.356)(105.0,0.792)(110.0,1.382)(115.0,2.111)(120.0,2.955)(125.0,3.888)(130.0,4.883)(135.0,5.909)(140.0,6.935)(145.0,7.930)(150.0,8.864)(155.0,9.707)(160.0,10.436)(165.0,11.027)(170.0,11.462)(175.0,11.728)(180.0,11.818)};
\node at (90,11.75) {$J_0>J_e$};
\node [above] at (180,0.182) {$J_0$};
\node [above] at (180,11.818) {$J_e$};
\node [rotate=90] at (\xlbl,7) {Assignment loss};
\end{scope}

\end{footnotesize}
		\end{tikzpicture}	
		\caption{Fusion and association loss functions \wrt $\alpha$. }	
		\label{fig:ex-motiv-results}	
	\end{subfigure}
	\caption{Motivating example. Two agents estimate two targets. By construction $\M\suba=\M\subb=\M$, where $\M(\alpha)=\BBM \cos\alpha & \sin\alpha \EBM$ and $\alpha\in[0\degrees,180\degrees]$. The dashed lines in (a) represent projections of the state estimates along $\M(0\degrees)$ and $\M(90\degrees)$. The effect of $\M$ on the fusion and association performance is evaluated by varying $\alpha$. The fusion loss function is $\trace(P\subi)$ and the association loss function is $\trace(\Pi\Ared)$, with $J_0$ and $J_e$ defined as the losses corresponding to correct and incorrect assignment, respectively. The fusion optimal $\M$ is given by $\alpha^\star=90\degrees$. At $\alpha^\star$, $\M y\aa=\M y\bb$ and $\M y\ab=\M\ba$ which implies $J_e=0<J_0$. }
\end{figure*}

\subsection{Association Model} \label{sec:asso-model}
The association problem is formulated as a linear assignment problem \cite{burkard2009assignment}. In case of full estimates, the assignment matrix is 
\begin{equation}
	\Afull = 
	\BBM 
		d_{(11)}^2 & \hdots & d_{(1N)}^2 \\
	 	\vdots & \ddots & \vdots \\
	 	d_{(N1)}^2 & \hdots & d_{(NN)}^2
	\EBM,
\end{equation}
where $d\subij^2$ is a {Mahalanobis distance} (MD) given by
\begin{subequations}
\begin{align}
	d\subij^2 &= \bary\subij\trnsp S\subij\inv \bary\subij, \\
	\bary\subij &= y\ai-y\bj, \\
	S\subij &= R\ai + R\bj,
\end{align}	
\end{subequations}
since $\EV(\bfv\ai\bfv\bj\trnsp)=0$. Similarly, the dimension-reduced assignment matrix $\Ared$ is defined as
\begin{equation}
	\Ared = 
	\BBM 
		r_{(11)}^2 & \hdots & r_{(1N)}^2 \\
		\vdots & \ddots & \vdots \\
		r_{(N1)}^2 & \hdots & r_{(NN)}^2  
 	\EBM,	
 	\label{eq:ass-mat-red}
\end{equation}
where $r\subij^2$ is an MD given by
\begin{align}
	r\subij^2 
	&= (\M\subj y\ai-y\Mj)\trnsp\left(\M\subj R\ai\Mt\subj+R\Mj\right)\inv \nonumber \\
	&\quad \times (\M\subj y\ai-y\Mj) \nonumber \\
	&= \bary\subij\trnsp \M\subj\trnsp \left(\M\subj S\subij\M\subj\trnsp\right)\inv \M\subj\bary\subij.
\end{align}

Agent~1 receives estimates from agent~2 and solves the association problem using the following optimization formulation. Let $\perms^N$ be the set of all $N\times N$ {permutation matrices}, \ie, 
\begin{equation}
	\perms^N = \left\{\Pi\in\reals^{N\times N} \,\left|\, [\Pi]_{ij}\in\{0,1\}, \Pi\Pi\trnsp = I  \right.\right\}.
\end{equation}
A permutation matrix $\Pi\in\perms^N$ assigns exactly one estimate in $\estset_1$ to each of the estimates in $\estset_\M$. The optimal $\Pi$ for a certain assignment matrix $\calA$ is computed using \cite{burkard2009assignment}
\begin{equation} \label{eq:asso-opt-prob}
\begin{aligned}
	& \underset{\Pi}{\minimize} & & \trace(\Pi\calA) \\
	& \subjectto & & \Pi \in \perms^N.
\end{aligned}	
\end{equation}
In this formulation correct assignment is given by $\Pi_0=I$. 

\begin{rmk}
Let $Z=[z_{ij}]$, where $z_{ij}\in\{0,1\}$. The problem in \eqref{eq:asso-opt-prob} is a matrix version of 
\begin{equation*}
\begin{aligned}
	& \underset{Z}{\minimize} & & z_{ij}[\calA]_{ij} \\
	& \subjectto & & \sum_{i} z_{ij} = 1, \quad \forall i, \\
	& & & \sum_{j} z_{ij} = 1, \quad \forall j.
\end{aligned}	
\end{equation*}
This formulation is more common in the MTT literature \cite{blackman:mts}. However, here we use \eqref{eq:asso-opt-prob}.
\end{rmk}


\subsection{Motivating Example} \label{sec:motiv-example}

We will now illustrate how the choice of $\M\subi$ affects the association performance. Consider the scenario in Fig.~\ref{fig:ex-motiv-scenario}, where $N=2$, $n=2$ and $m=1$. Each agent has a local estimate of each of the two targets as defined in Fig.~\ref{fig:ex-motiv-scenario}, where $R\aa=R\ab$ and $R\ba=R\bb$. Assume 
\begin{equation*}
	\M\suba=\M\subb=\M=\BBM\cos\alpha&\sin\alpha\EBM,
\end{equation*}
where $\alpha$ is an angle. Based on this parametrization it is possible to define $\Ared$ as a function of $\alpha$. Let 
\begin{align*}
	J_0 &= \left.\trace(\Pi\Ared)\right|_{\Pi=\Pi_0} = r_{(11)}^2+r_{(22)}^2, \\
	J_e &= \left.\trace(\Pi\Ared)\right|_{\Pi=\BBSM0&1\\1&0\EBSM} = r_{(12)}^2+r_{(21)}^2,
\end{align*}
be the cost corresponding to correct and incorrect assignment, respectively. By construction $J_0$, $J_e$ and $\trace(P\suba)=\trace(P\subb)$ are functions of $\alpha$. 

The fusion and association performance with respect to (\wrt) $\alpha$ is evaluated by computing $J_0$, $J_e$ and $\trace(P\subi)$ for each $\alpha\in[0\degrees,180\degrees]$. The results are shown in Fig.~\ref{fig:ex-motiv-results}. The fusion optimal $\M$ corresponds to $\alpha^\star=90\degrees$. However, this $\M$ lies in the interval where $J_0>J_e$ which would imply incorrect assignment. To have correct assignment in the dimension-reduced case while maintaining good fusion performance the selected $\M$ should be such that it minimizes $\trace(P\subi)$ subject to $J_0<J_e$.

\subsection{Problem Formalization}
Assume the targets $x\suba,\dots,x\subN$ are well separated such that solving the assignment problem in \eqref{eq:asso-opt-prob} with $\calA=\Afull$ yields $\Pi_0$. Moreover, assume that agent~2 has no knowledge about $\estset_1$. The problem is, at agent~2, to compute $\M\suba,\dots,\M\subN\in\reals^{1\times n}$ such that when agent~1 solves \eqref{eq:asso-opt-prob} with $\calA=\Ared$ the solution $\Pi$ is as close as possible to $\Pi_0$. In other words, since it in general is not possible to obtain correct association in the dimension-reduced case, we want to compute $\M\suba,\dots,\M\subN$ in such a way that the association is not degraded too much. The focus is on the case $m=1$. However, some of the results are given for arbitrary $m\geq1$.	

\begin{rmk}
The considered problem is not the common association problem of network-centric MTT where received tracks are associated with local tracks and correct assignment $\Pi_0$ is \emph{unknown}. Here, the correct assignment is \emph{known} by construction and hence, for the presentation, we have the freedom of defining $\Afull$ and $\Ared$ such that $\Pi_0=I$. 
\end{rmk}

\section{Problem Analysis} \label{sec:analysis}
In this section we examine properties of the considered association problem. Sufficient conditions for correct assignment are given. An example is used to show that the problem is further complicated by inherent randomness. Statistical properties of the problem are derived at the end to be used in the subsequent section.


\subsection{A Sufficient Condition for Correct Assignment}
Consider now an oracle's perspective. The example of Sec.~\ref{sec:motiv-example} illustrate an important property of the problem. That is, for $\M\subj\neq 0$ and $\bary\subij\neq0$
\begin{equation*}
	\M\subj \perp \bary\subij\trnsp \iff \M\subj\bary\subij = 0, 
\end{equation*}
where $\bary\subij=y\ai-y\bj$. From this it can be inferred that for the association we want
\begin{equation}
	\M\subj\bary\subjj = 0 \quad \wedge \quad i\neq j \implies \M\subj\bary\subij \neq 0,
	\label{eq:suff-cond}
\end{equation}
where $\wedge$ is {logical and}, since in this case $r\subjj^2=0$ and $r\subij^2>0$ if $i\neq j$. A sufficient condition for correct assignment is hence that \eqref{eq:suff-cond} holds for all $j$ as this would imply $\trace(\Ared)=0$. However, by assumption agent~2 has no knowledge about $\estset_1$ and hence without further knowledge agent~2 cannot compute $\M\subj$ such that \eqref{eq:suff-cond} is satisfied.

\subsection{Problem Properties}
In the example of the previous section the fusion optimal $\M$ gave incorrect association. Luckily, it is not generally the case that the fusion optimal $\M$ yields incorrect assignments. Unfortunately, it is impossible to say something general about tradeoffs between fusion and association performance. The main reasons for this are described below.

Consider $\Mt\subj\in\realsn$, and let $Q\subj=R\aj^2\in\pdsetn$ and $S\subjj=R\aj+R\bj\in\pdsetn$. In the fusion case the optimal $\M\subj$ solves \cite{forsling2022fusion}
\begin{equation} \label{eq:gevo-kf-m=1}
\begin{aligned}
	& \underset{\|\M\subj\|=1}{\maximize} & & \frac{\M\subj Q\subj\Mt\subj}{\M\subj S\subjj\Mt\subj}.
\end{aligned}	
\end{equation}
Hence the fusion optimal $\M\subj$ for a certain target $x\subj$ can be solved isolated from the other targets. This is not true in the association problem where optimal $\M\subj$ for a certain target $x\subj$ depends on all estimates in both $\estset_1$ and $\estset_2$ through $\Ared$.

A slightly less restrictive sufficient condition for correct assignment, \cf \eqref{eq:suff-cond}, is that for each $j$
\begin{equation}
	r\subjj^2 < r\subij^2, \quad \forall i\neq j.
	\label{eq:suff-cond-2}
\end{equation} 
If this condition holds {nearest neighbor} \cite{blackman:mts} association yields the same results as GNN association. The condition in \eqref{eq:suff-cond-2} can also be expressed as \cite{forsling2022fusion}
\begin{equation}
	\frac{\M\subj\bary\subjj\bary\subjj\trnsp\Mt\subj}{\M\subj S\subjj \Mt\subj} < \frac{\M\subj\bary\subij\bary\subij\trnsp\Mt\subj}{\M\subj S\subij \Mt\subj}, \quad \forall i\neq j,
\end{equation}
where each fraction is structurally similar to the fraction in \eqref{eq:gevo-kf-m=1}. However, a complication compared to the fusion case is that $r\subij^2$ is a realization of a random variable
\begin{equation*}
	\bfr\subij^2 = \bbfy\subij\trnsp\Mt\subj \left(\M\subj S\subij\Mt\subj \right)\inv \M\subj\bbfy\subij,
\end{equation*}
where $\bbfy\subij=\bfy\ai-\bfy\bj$. Hence, assuming that agent~2 has access to $R\ai$ and a good estimate of $x\subi$, the fusion optimal $\M\subi$ could be computed while it would still be difficult to predict $r\subij^2$ due to randomness. Fig.~\ref{fig:example-realization} shows two possible realizations of each of the random variables
\begin{align*}
	\bfy\aa &= x\suba + \bfv\aa, & \bfy\ba &= x\suba + \bfv\ba, \\
	\bfy\ab &= x\subb + \bfv\ab, & \bfy\bb &= x\subb + \bfv\bb,
\end{align*}
where $\bfv\aa,\bfv\ab\sim\calN(0,R_1)$ and $\bfv\ba,\bfv\bb\sim\calN(0,R_2)$. Since the covariances are the same in each case and since by assumption $R\aa=R\ab=R_1$ and $R\ba=R\bb=R_2$ we have that fusion optimal $\M\subj$ satisfy $\M\suba=\M\subb=\M$ in both cases. Computing $\Ared(\M)$ in realization~1 and realization~2 yields
\begin{align*}
	\calA_1 &= \BBM0.05&1.01\\0.31&0.05\EBM, & \calA_2 &= \BBM0.11&0.01\\0.01&0.11\EBM,
\end{align*}
respectively. In realization~1 we will hence have correct assignment $\Pi_0$ while in realization~2 the incorrect combination is chosen. The example illustrates that, due to the inherent randomness, it is in general impossible to decide if a fusion optimal $\M\subj$ will imply correct or incorrect assignment without knowing the actual realization.

%

\begin{figure}[tb]
	\centering
	\begin{tikzpicture}[scale=0.30]
		\input{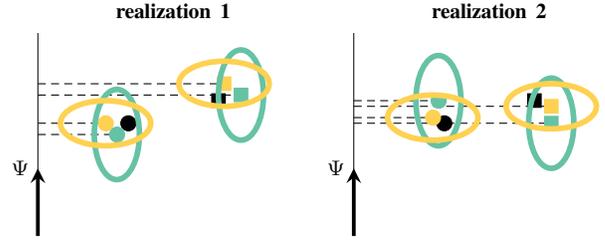}
	\end{tikzpicture}
	\caption{Two noise realizations of the same scenario. The target states and the covariances $R\aa=R\ab$ and $R\ba=R\bb$ are the same in both realizations. What differs are the state estimates $y\aa$, $y\ba$, $y\ab$ and $y\bb$. In realization~1 correct assignment is obtained while in realization~2 incorrect assignment is obtained.}
	\label{fig:example-realization}
\end{figure}

\subsection{Statistical Properties}
Assume $m\geq1$. By construction
\begin{equation*}
	\M\subj\bbfy\subij\sim \calN\left(\M\subj\barx\subij,\M\subj S\subij\Mt\subj\right),
\end{equation*} 
where $\barx\subij=x\subi-x\subj$. Hence \cite{rao73}
\begin{equation}
	\bfr\subij^2 \sim 
	\begin{cases}
		\chi_m^2, & \text{ if } i=j, \\
		\chi_{m,\nu}^2, & \text{ if } i\neq j,
	\end{cases}
\end{equation}
where $\chi_m^2$ is the central chi-squared distribution with $m$ degrees of freedom, and $\chi_{m,\nu}^2$ is the noncentral chi-squared distribution, where $\nu$ is the noncentrality parameter. The expectation value is 
\begin{align}
	\EV\left(\bfr\subij^2\right) = m + \nu\subij,
	\label{eq:r2-ev}
\end{align}
where $\nu\subij=\barx\subij\trnsp \Mt\subj \left(\M\subj S\subij\Mt\subj\right)\inv\M\subj\barx\subij$ is the noncentrality parameter. The variance is given by \cite{rao73}
\begin{align}
	\var\left(\bfr\subij^2\right) = 2m + 4\nu\subij.
	\label{eq:r2-var}
\end{align}
One conclusion is that as $\nu\subij$ increases the relative effect of randomness decreases since $\EV(\bfr\subij^2)$ scales as $\nu\subij$ while $\sqrt{\var(\bfr\subij^2)}$ only scales as $\sqrt{\nu\subij}$. This result is important and is used in the solution proposed in the next section.

\section{Preserving Correct Assignment With Dimension-Reduced Estimates} \label{sec:solution} 

In this section a method for preserving high association quality is suggested. Based on the analysis of Sec.~\ref{sec:analysis}, an optimization formulation is provided for computation of $\M\subj$. This leads to the proposed descent based optimization strategy, where a key ingredient and contribution is an adaptive step size. At the end we provide a numerical example and some comments about the optimization strategy.

\subsection{Approximated Assignment Matrix}
The proposed solution is based on the analysis of the previous section. In particular, we estimate $r\subij^2$ using $\EV(\bfr\subij^2)$ in \eqref{eq:r2-ev}. To compute $r\subij^2$ agent~2 must have access to both $(y\ai,R\ai)$ and $(y\bj,R\bj)$, but $(y\ai,R\ai)$ is unknown to agent~2. An approximation to $(y\ai,R\ai)$ which is already locally available is $(y\bi,R\bi)$. Let
\begin{equation}
	\hatr\subij^2 = \haty\subij\trnsp\Mt\subj\left(\M\subj\hatS\subij\Mt\subj\right)\inv\M\subj\haty\subij,
\end{equation}
where $\haty\subij=y\bi-y\bj$ and $\hatS\subij=R\bi+R\bj$ such that $\hbfy\subij\sim\calN(\barx\subij,\hatS\subij)$. This is consistent with $\bfr\subij^2$ in the sense that
\begin{align*}
	&\EV\left(\hbfr\subij^2\right) \\
	&\quad 
	\begin{small}
	=
	\begin{cases}
		m, & \text{ if } i=j, \\
		m + \barx\subij\trnsp\Mt\subj\left(\M\subj\hatS\subij\Mt\subj\right)\inv\M\subj\barx\subij, & \text{ if } i\neq j.
	\end{cases}
	\end{small}
\end{align*}
which is identical to \eqref{eq:r2-ev} except that $S\subij$ is replaced by $\hatS\subij$. We then define the \emph{approximated assignment matrix} as
\begin{equation}
	\Aredap = 
	\BBM 
		\hatr\subaa^2 & \hdots & \hatr\subaN^2 \\
		\vdots & \ddots & \vdots \\
		\hatr\subNa^2 & \hdots & \hatr\subNN^2
	\EBM.
	\label{eq:ass-mat-red-ap}
\end{equation}


\subsection{Proposed Solution}

Since we only have access to an approximation $\Aredap$ of $\Ared$, $\M\subj$ is computed based on the sufficient condition in Sec.~\ref{sec:analysis}, \cf \eqref{eq:suff-cond-2}. The condition is utilized because we want to have some marginal when choosing $\M\subj$, to avoid that $r\subij^2$ is zero or very small if $i\neq j$. Moreover, if $\M\subj$ satisfies this sufficient condition there is no need to take into account the other $\M\subi,i\neq j$ when computing $\M\subj$---correct assignment is obtained regardlessly. 

Consider now a certain $j$ and $\M\subj$. Let
\begin{align}
	f_i(z) &= \frac{z\trnsp \hatY\subij z}{z\trnsp \hatS\subij z}, & \hatY\subij &= \haty\subij\haty\subij\trnsp,
\end{align}
be defined for all $i\neq j$. To maximize $f_i(z)$ simultaneously for all $i$ is in general impossible since this is a \emph{multiobjective} optimization problem. However, we can consider a worst-case approach and maximize the minimum $f_i(z)$. This implies a \emph{maximin} formulation where $\M\subj$ is computed using
\begin{equation} \label{eq:maximin-prob}
\begin{aligned}
	&\underset{\M\subj}{\maximize} & & \left( \underset{i\neq j}{\min}\quad f_i(\Mt\subj) \right).
\end{aligned}
\end{equation}
The problem in \eqref{eq:maximin-prob} is a \emph{nonconvex problem} involving optimization over a finite set of quadratic form ratios. The problem is difficult to solve in general and therefore the following optimization strategy is proposed.

\subsection{Optimization Strategy}
For each individual $f_i(z)$, the $z$ that maximizes $f_i(z)$ is known to be given by the eigenvector $u$ that corresponds to the maximum eigenvalue $\lambda$ of \cite{forsling2022fusion}
\begin{equation}
	\hatY\subij u = \lambda \hatS\subij u.
\end{equation}
As $\hatY\subij\in\psdsetn$ and $\rank(\hatY\subij)=1$ this eigenvalue problem has only one strictly positive eigenvalue $\lambda$ for which the corresponding eigenvector is denoted by $\umaxij$. Since $\umaxij$ in general differ for different $i$, it is not possible to maximize all $f_i(z)$ simultaneously. However, for a certain $z$ we know the values of all $f_i(z)$ and hence are able to compute
\begin{equation}
	 \imin = \underset{i\neq j}{\argmin}\quad f_i(z).
	 \label{eq:imin}
\end{equation}
To increase $f_{\imin}$ it is suggested that  
\begin{equation}
	z \leftarrow z + \alpha \umax_{\imin},
	\label{eq:z-new}
\end{equation}
where $\alpha$ resembles the step size to traverse along $\umax_{\imin}$. Using a too large $|\alpha|$ there is a risk that $f_i$ for some other $i\neq \imin$ is severely decreased. Too small $|\alpha|$ means slow convergence. From Proposition~\ref{prop:directional-derivative} we have that a first-order approximation of $f_{i}$ evaluated at $z$ in the direction of $\alpha\umax_{\imin}$ is given by
\begin{equation}
	f_i(z+\alpha\umax_{\imin}) \approx f_i(z) + 2\alpha\frac{\umax_{\imin}\trnsp(\hatY\subij-f_i(z)\hatS\subij)z}{z\trnsp\hatS\subij z}.
	\label{eq:first-order-approx}
\end{equation}
We proceed by solving
\begin{align}
	&f_i(z) + 2\alpha\frac{\umax_{\imin}\trnsp(\hatY\subij-f_i(z)\hatS\subij)z}{z\trnsp \hatS\subij z} \nonumber \\
	&\quad = 
	f_{\imin}(z) + 2\alpha\frac{\umax_{\imin}\trnsp(\hatY_{\imin j}-f_i(z)\hatS_{\imin j})z}{z\trnsp \hatS_{\imin j} z},
\end{align}
for each $i\neq j,\imin$. This yields $N-2$ solutions for $\alpha$, where some might be negative and other positive. Since the task is to increase $f_{\imin}$ while not decreasing the other $f_i$ too much, $\alpha$ is chosen such that $|\alpha|$ is the smallest among all the ones that satisfy 
\begin{equation}
	\alpha \frac{\umax_{\imin}\trnsp(\hatY_{\imin j}-f_i(z)\hatS_{\imin j})z}{z\trnsp \hatS_{\imin j} z} > 0.
	\label{eq:alpha-equation}
\end{equation}
This last condition is introduced to ensure that the correct sign is chosen for $\alpha$. 


The operations in \eqref{eq:imin}--\eqref{eq:alpha-equation} are performed iteratively until some termination criterion is met. The optimization algorithm is summarized in Algorithm~\ref{alg:asso-psi}.

\begin{prop} \label{prop:directional-derivative}
Let $u,z\in\realsn$, $Y,S\in\realsnn$ and $f(z) = (z\trnsp Yz)/(z\trnsp Sz)$, where $z\neq0$ and $\rank(S)=n$. Then a first-order approximation of $f(z+\alpha u)$, for any scalar $\alpha$, is given by
\begin{equation}
	f(z+\alpha u) \approx f(z) + 2\alpha\frac{u\trnsp(Y-f(z)S)z}{z\trnsp Sz}.
	\nonumber
\end{equation}
\end{prop}

\begin{proof}
From \cite{matrix-cookbook2012} we have
\begin{align*}
	\frac{\partial f(z)}{\partial z} 
	&= -\frac{2Szz\trnsp Yz}{(z\trnsp Sz)^2} + \frac{2Yz}{z\trnsp Sz} = -\frac{2f(z)Sz}{z\trnsp Sz} + \frac{2Yz}{z\trnsp Sz} \\ 
	&= 2\frac{(Y-f(z)S)}{z\trnsp Sz}z.
\end{align*}
A first-order approximation of $f(z+\alpha u)$ is given by
\begin{align*}
	f(z+\alpha u) 
	&\approx 
	f(z) + \left.\alpha u\trnsp \frac{\partial f(z')}{\partial z'}\right\vert_{z'=z} \\
	&= f(z) + 2\alpha\frac{u\trnsp(Y-f(z)S)z}{z\trnsp Sz}.
\end{align*}
\end{proof}

\begin{algorithm}[t]
	\caption{Association Quality Based $\M\subj$}
	\label{alg:asso-psi}
	\begin{algorithmic}[0]
		\begin{small}
		\Input $\estset_2$, $j$, $\alphalow$ and $\alphahigh$ \\
		\begin{enumerate}[1:]
			\item For each $i\neq j$: Let $\hatY\subij = \haty\subij\haty\subij\trnsp$, $\hatS\subij=R\bi+R\bj$ and $f_i(z)=(z\trnsp\hatY\subij z)/(z\trnsp\hatS\subij z)$. Let $\umaxij$ be the eigenvector corresponding to the maximum eigenvalue $\lambda_i$ of $\hatY\subij u = \lambda\hatS\subij u$. 
			\item Let $k=0$ and $z_0\leftarrow z_0/\|z_0\|$, where $z_0=\sum_{i=1,i\neq j}^N \frac{1}{\lambda_i}u_i$.
			\item Let $k\leftarrow k+1$. Compute
				\begin{equation*}
					\imin = \underset{i\neq j}{\argmin}\quad f_i(z_{k-1}).
				\end{equation*}
			\item For each $i\neq j$: Define
				\begin{align*}
					&\hatf_i(z_{k-1}+\alpha u_{\imin}) \\ 
					&\quad = f_i(z_{k-1}) + 2\alpha\frac{\umax_{\imin}\trnsp(\hatY\subij-f_i(z_{k-1})\hatS\subij)z_{k-1}}{z_{k-1}\trnsp \hatS\subij z_{k-1}}.
				\end{align*}
			\item For each $i\neq j,\imin$: Solve for $\alpha$ using $\hatf_i=\hatf_{\imin}$. Store the different $\alpha$ in a vector $a$.
			\item If 
				\begin{equation*}
					\alpha \frac{\umax_{\imin}\trnsp(\hatY_{\imin j}-f_i(z_{k-1})\hatS_{\imin j})z_{k-1}}{z_{k-1}\trnsp \hatS_{\imin j} z_{k-1}} > 0,
				\end{equation*}
				then let $\alpha_k$ be given by the minimum positive element of $a$. Otherwise, let $\alpha_k$ be given by the maximum negative element of $a$. If $|\alpha_k|<\alphalow$, then let $\alpha_k\leftarrow\sign(\alpha_k)\alphalow$. If $|\alpha_k|>\alphahigh$, then let $\alpha_k\leftarrow\sign(\alpha_k)\alphahigh$.
			\item Let $z_k\leftarrow z_k/\|z_k\|$, where $z_k=z_{k-1}+\alpha_k\umax_{\imin}$.
			\item Terminate with $\M\subj=z_k\trnsp$ if a predefined stopping criterion is met. Otherwise, go back to step~3.
		\end{enumerate}
		\Output $\M\subj$
		\end{small}
	\end{algorithmic}
\end{algorithm}

\subsection{Example}
As an example of the proposed optimization strategy, consider a scenario with $N=3$ and $n=4$. Assume $j=3$. Since $N=3$ we consider two loss functions
\begin{align*}
	f_1(z) &= \frac{z\trnsp\hatY_{(13)}z}{z\trnsp\hatS_{(13)}z}, & f_2(z) &= \frac{z\trnsp\hatY_{(23)}z}{z\trnsp\hatS_{(23)}z}.
\end{align*}
The multiobjective problem of maximizing $f_1$ and $f_2$ simultaneously is not solvable, hence we will use the maximin approach and Algorithm~\ref{alg:asso-psi}. The original Algorithm~\ref{alg:asso-psi} uses an adaptive step size $\alpha\in[\alphalow,\alphahigh]$. We will compare this to the same algorithm with: (i) a small fixed step size $\alpha=\alphalow$, and (ii) a large fixed step size $\alpha=\alphahigh$.

The optimization results for the three cases, which all use the same initial vector $z_0$, are shown in Fig.~\ref{fig:opt-alg-analysis} for $\kmax=25$ iterations. In Fig.~\ref{fig:opt-alg-analysis-pareto} $f_1$ is plotted against $f_2$. The yellow dots resemble $f_1$ and $f_2$ at randomly sampled $z$. Fig.~\ref{fig:opt-alg-analysis-results} visualizes
\begin{equation*}
	\fmin = \min\, \left(f_1,f_2\right),
\end{equation*}
for each iteration $k=1,2,\dots,\kmax$. In this case the adaptive step size provides the best results. The small step size gives slow convergence while the large step oscillates as it becomes inaccurate due to the large step size. It cannot be concluded if Algorithm~\ref{alg:asso-psi} have reached a global maximum or a stationary point. 

\begin{figure}[tb]
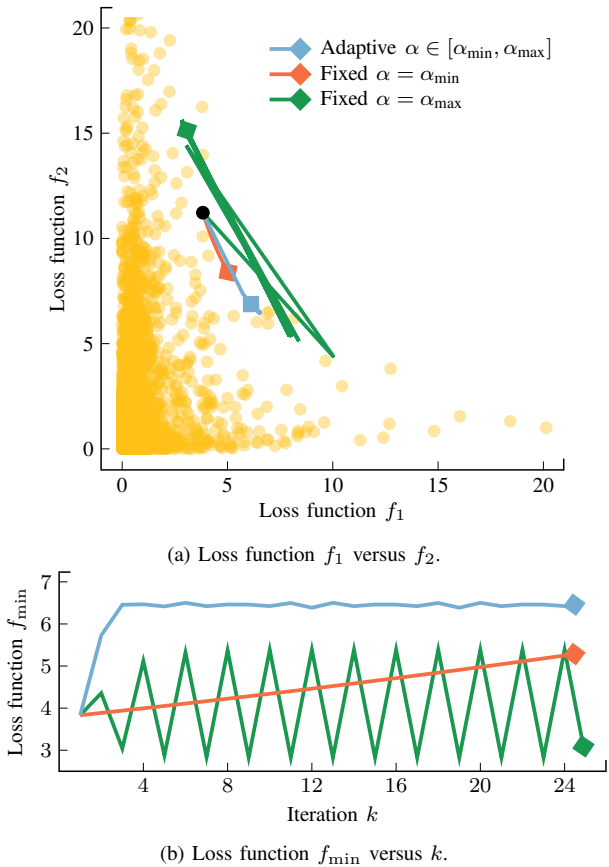

	\centering
	\begin{subfigure}[b]{1.0\columnwidth}
	\centering
		\begin{tikzpicture}[scale=.8]
			\input{fig/pareto-analysis.tex}
		\end{tikzpicture}
		\caption{Loss function $f_1$ versus $f_2$.}
		\label{fig:opt-alg-analysis-pareto}	
	\end{subfigure}
	\begin{subfigure}[b]{1.0\columnwidth}
		\centering
		\begin{tikzpicture}
			\input{fig/optimization-algorthim-analysis.tex}
		\end{tikzpicture}	
		\caption{Loss function $\fmin$ versus $k$.}		
		\label{fig:opt-alg-analysis-results}
	\end{subfigure}
	\caption{Example of the proposed optimization strategy with $N=3$ and $n=4$. Algorithm~\ref{alg:asso-psi} is compared to the same algorithm but with fixed step size. The black circle marks the common initial value $z_0$. The squares mark the final point of each case. Yellow dots resemble $f_1$ and $f_2$ evaluated at randomly generated $z$. A small step size yields slow convergence. A large step size yields inaccurate results but possibly a higher convergence rate. The adaptive step size outperforms the other two.}
	\label{fig:opt-alg-analysis}
\end{figure}

\subsection{Comments}
In essence the proposed optimization strategy in Algorithm~\ref{alg:asso-psi} is an iterative descent based optimization method, where the descent directions are chosen from a finite set of predefined directions. In this interpretation step~4--6 correspond to a backtracking line search where the step size $\alpha$ is selected. Algorithm~\ref{alg:asso-psi} takes $\alphalow>0$ as an input to avoid getting stuck at local minima, and $\alphahigh>\alphalow$ such that the linear approximation given by \eqref{eq:first-order-approx} does not become too poor. The stopping criterion used in this paper is $k>\kmax$, \ie, the algorithm terminates after $\kmax$ iterations. 

It is possible to include more sophisticated optimization techniques for better performance, but such techniques are out of the scope in this paper. It should be emphasized that there are no guarantees that Algorithm~\ref{alg:asso-psi} converges to a global maximum \wrt the problem in \eqref{eq:maximin-prob}. In fact, simulations verify that in general only local maxima are reached.



\section{Numerical Evaluation} \label{sec:evaluation}
In this section we provide a numerical evaluation of Algorithm~\ref{alg:asso-psi}. The association performance when computing $\M\subj$ using Algorithm~\ref{alg:asso-psi} is compared to the case when $\M\subj$ is computed using Algorithm~\ref{alg:gevo-kf}.

\subsection{Simulation Specification}
A target tracking scenario with $N=10$ targets is assumed. It is assumed that the dimensionality $n=6$ which we here interpret as a constant acceleration model in two spatial dimensions \cite{li-jilkov:dynamic-models}. For each target $x\subi$ a pair of covariances $R\ai$ and $R\bi$ are defined and are held fixed throughout the simulations. A \emph{Monte Carlo} (MC) approach is used, where in each MC run the state estimates $y\ai$ and $y\bi$ are sampled using $R\ai$ and $R\bi$, respectively, and the model in \eqref{eq:est-model}. We also use a \emph{scaling factor} $c$ that scales the two spatial uncertainty components. Hence, for larger $c$ the association problem becomes more difficult to solve, and for smaller $c$ the association problem becomes easier to solve. The assumed target tracking scenario is depicted in Fig.~\ref{fig:scenario} with $c=1$.

\begin{figure}[tb]
	\centering
	\begin{tikzpicture}[scale=.33]
		\newcommand*\targetoption{[black]}

\def\rtgt{0.25}

\def\xleg{13}
\def\yleg{10}
\def\sleg{1.2}
\def\lleg{1.0}

\input{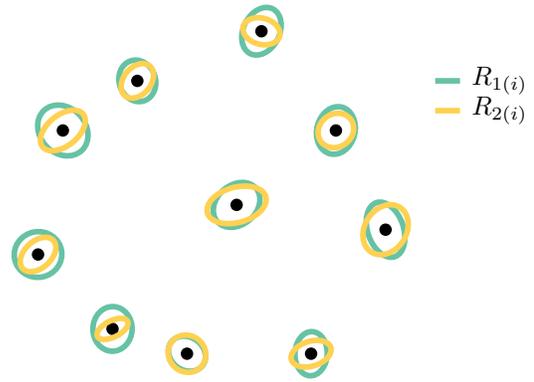}

\draw\ellipseoption{colora} ({\xleg},{\yleg})--({\xleg+\lleg},{\yleg}) node [right,black] {$R\ai$};
\draw\ellipseoption{colorb} ({\xleg},{\yleg-\sleg})--({\xleg+\lleg},{\yleg-\sleg}) node [right,black] {$R\bi$};
	\end{tikzpicture}
	\caption{Numerical scenario. The targets are represented by black dots. The ellipse around a target illustrates the uncertainty of the corresponding estimate in the two spatial dimensions. }
	\label{fig:scenario}
\end{figure}

To evaluate association performance the \emph{incorrect assignment rate} $\PIC$ is computed for a certain $c$ as the mean over all MC runs of the number of incorrect assignments divided by $N$. We compute $\PIC$ for the following cases:
\begin{itemize}
	\item $\left(\estset_1,\estset_2\right)$: The {full estimate} configuration where agent~1 receives receives $\estset_2$ from agent~2.
	\item $\left(\estset_1,\estset_\M\right)$ + $\M\subj$ using Alg.~\ref{alg:gevo-kf}: A dimension-reduced configuration where agent~1 receives $\estset_\M$ from agent~2 and $\M\subj$ is computed using Algorithm~\ref{alg:gevo-kf}. In this case it is assumed that agent~2 has access to $\estset_1$ such that fusion optimal $\M\subj$ can be computed.
	\item $\left(\estset_1,\estset_\M\right)$ + $\M\subj$ using Alg.~\ref{alg:asso-psi}: A dimension-reduced configuration where agent~1 receives $\estset_\M$ from agent~2 and $\M\subj$ is computed using the proposed optimization strategy of Algorithm~\ref{alg:asso-psi}.
\end{itemize}
The standard deviation of $\PIC$ is also computed. 

\begin{rmk}
Since agent~1 needs $\M\suba,\dots,\M\subN$ to be able to fuse the estimates in $\estset_\M$ with its local estimates, agent~2 must also include $\M\suba,\dots,\M\subN$ when transmitting	 $\estset_\M$. Functionality for encoding $\M\subj$ is described in \cite{forsling23tsp-gevo} with \matlab code available at \url{https://gitlab.com/robinforsling/dtt/}.
\end{rmk}

\subsection{Results}
The results of the numerical evaluation are visualized in Fig.~\ref{fig:results}, where $\PIC$ is plotted against $c$. For each value of $c$, $M=1\,000$ MC runs are evaluated. The quantity $\PIC$ is computed in the same realizations of $\estset_1$ and $\estset_2$ for each of the cases described previously. The shaded areas in the plot resemble 1-$\sigma$ confidence intervals.

Perfect association is maintained in the full estimate case for all values of $c$. The approach that utilizes Algorithm~\ref{alg:asso-psi} clearly outperforms the approach that computes $\M\subj$ for optimal fusion performance.

\begin{figure}[tb]
	\centering
	\begin{tikzpicture}
\def\xmin{-0.1}; \def\xmax{5.2}
\def\ymin{-0.05}; \def\ymax{1.05}

\def\xleg{0.5}
\def\yleg{1.15}
\def\sleg{0.075}
\def\lleg{0.35}

\def\xlbl{0.6}
\def\ylbl{0.1}

\def\xtick{0.1}
\def\ytick{0.025}

\begin{footnotesize}
\begin{scope}[xscale=1.4,yscale=5.6]

\foreach \y in {0,0.25,0.5,0.75,1} {
	\draw\tickoption({\xmin+\xtick},\y)--(\xmin,\y) node [left] {$\y$};
}
\foreach \x in {1,...,5} {
	\draw\tickoption(\x,{\ymin+\ytick})--(\x,\ymin) node [below] {$\x$};
}

\draw\coordframeopt (\xmin,\ymax)--(\xmin,\ymin)--(\xmax,\ymin);
\node at ({0.5*\xmin+0.5*\xmax},{\ymin-\ylbl}) {Scaling factor $c$};
\node[rotate=90] at ({\xmin-\xlbl},{0.5*\ymin+0.5*\ymax}) {Incorrect assignment rate $\PIC$};

\input{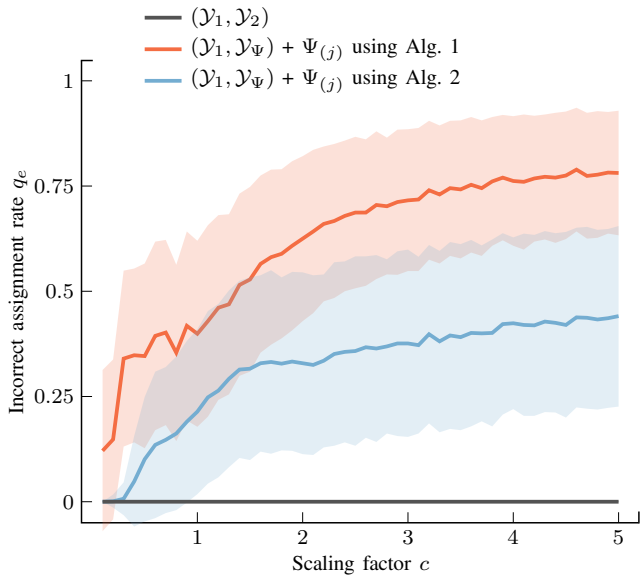}

\draw\meanoption{clrfull} ({\xleg},{\yleg})--({\xleg+\lleg},{\yleg}) node [right,black] {$\left(\estset_1,\estset_2\right)$};
\draw\meanoption{clrfus} ({\xleg},{\yleg-\sleg})--({\xleg+\lleg},{\yleg-\sleg}) node [right,black] {$\left(\estset_1,\estset_\M\right)$ + $\M\subj$ using Alg.~\ref{alg:gevo-kf}};
\draw\meanoption{clrasso} ({\xleg},{\yleg-2*\sleg})--({\xleg+\lleg},{\yleg-2*\sleg}) node [right,black] {$\left(\estset_1,\estset_\M\right)$ + $\M\subj$ using Alg.~\ref{alg:asso-psi}};

\end{scope}
\end{footnotesize}
	\end{tikzpicture}
	\caption{Results of the numerical evaluation. The incorrect assignment rate $\PIC$ is computed as a sample mean for each of the three cases for different values of $c\in[0.1,5.0]$. The shaded areas illustrates the standard deviation of $\PIC$.}
	\label{fig:results}
\end{figure}

\section{Concluding Remarks} \label{sec:conc}
The association problem for multitarget tracking in a dimension-reduced context has been proposed. In it, the track estimates to be communicated from one agent are dimension-reduced with respect to association quality in the agent that receives the dimension-reduced estimates. The implied problem was analyzed theoretically where it was illustrated that the problem is versatile and complex, and where no general solutions exists. An \emph{optimization strategy} has been suggested for computing dimension-reduced estimates while preserving \emph{association performance}. The optimization strategy was demonstrated using a numerical evaluation in which the suggested method outperformed a method that reduces dimensionality based on optimal \emph{fusion performance}.

Possible future extensions include a generalization of Algorithm~\ref{alg:asso-psi} for the $m>1$ case, and a more general configuration where agents have different sets of tracks. Another possibility is to consider a setting where there is partly knowledge available about the local estimates of the agent that receives the dimension-reduced estimates. A joint problem formulation which includes both fusion and association performance simultaneously is also of interest.

\bibliographystyle{IEEEtran}
\bibliography{bib/myrefs}

\end{document}